\documentclass[aps,prl,showpacs,twocolumn,nofootinbib,10pt]{revtex4-1}

\usepackage{graphicx}
\usepackage{amsmath}
\usepackage{amssymb}
\usepackage{mathrsfs}
\usepackage{amsthm}
\usepackage{bm}
\usepackage{url}
\usepackage[T1]{fontenc}
\usepackage{csquotes}
\MakeOuterQuote{"}

\newtheoremstyle{note}
  {\topsep/2}               
  {\topsep/2}               
  {}                      
  {\parindent}            
  {\itshape}              
  {.}                     
  {5pt plus 1pt minus 1pt}
  {}

\theoremstyle{note}
\newtheorem{theorem}{Theorem}
\newtheorem{lemma}{Lemma}

\theoremstyle{definition}

\theoremstyle{remark}
\newtheorem{remark}{Remark}

\newcommand{\mrm}[1]{\mathrm{#1}}
\providecommand{\tr}{\operatorname{tr}}

\newcommand{\diag}{\operatorname{diag}}

\newcommand{\rep}{\mathrel{\widehat{=}}}

\providecommand{\rmi}{\mathrm{i}}

\newcommand{\rmT}{\mathrm{T}}

\newcommand{\bbF}{\mathbb{F}}

\newcommand{\be}{\begin{equation}}
\newcommand{\ee}{\end{equation}}
\newcommand{\ba}{\begin{align}}
\newcommand{\ea}{\end{align}}

\def\<{\langle}  
\def\>{\rangle}  

\newcommand{\sy}{\bigl(\begin{smallmatrix}0 &-\rmi\\
\rmi &0\end{smallmatrix}\bigr)}

\newcommand{\Sp}[2]{\mrm{Sp}(#1,#2)}

\newcommand{\ASp}[2]{\mrm{ASp}(#1,#2)}
\newcommand{\SL}[2]{\mrm{SL}(#1,#2)}

\newcommand{\hw}{D}

\newcommand{\phw}{\overline{D}}

\newcommand{\Cli}{\mathrm{C}}

\newcommand{\pc}{\overline{\mathrm{C}}}

\def\eqref#1{\textup{(\ref{#1})}}  
\newcommand{\eref}[1]{Eq.~\textup{(\ref{#1})}}

\newcommand{\thref}[1]{Theorem~\ref{#1}}

\newcommand{\lref}[1]{Lemma~\ref{#1}}

\newcommand{\cref}[1]{Conjecture~\ref{#1}}
\newcommand{\Cref}[1]{Conjecture~\ref{#1}}

\begin{document}
	\title{Sharply covariant  mutually unbiased bases}
	\author{Huangjun Zhu}
	\email{hzhu@pitp.ca}
	\affiliation{Perimeter Institute for Theoretical Physics, Waterloo, Ontario N2L 2Y5, Canada}
	



\begin{abstract}

Mutually unbiased bases (MUB) are an elusive discrete structure in  Hilbert spaces.   Many (complete sets of) MUB are group covariant, but little is known whether they can be sharply covariant in the sense that the  generating groups can  have order equal to   the  total number of basis states, that is, $d(d+1)$ for MUB in dimension $d$.
Sharply covariant MUB, if they exist,  would be most appealing  from both theoretical and practical point of view. Since stabilizer MUB subsume almost all MUB that have ever been constructed, it is of fundamental interest  to single out those candidates that are sharply covariant.
We show that, quite surprisingly,  only two stabilizer MUB are sharply covariant, and the conclusion still holds even if antiunitary transformations are taken into account.
Our study provides valuable insight on the symmetry of stabilizer MUB, which may have implications for a number of research topics in quantum information and foundations. In addition, it  exposes a sharp contrast between MUB and another elusive discrete structure known as  symmetric informationally complete measurements (SICs), all known examples of which are sharply covariant.
\end{abstract}
\date{\today}
\maketitle

The existence of complementary observables is a main distinction between quantum physics and classical physics. Two observables are  complementary if, given the outcome of one observable, the outcome of the other  is maximally uncertain \cite{Bohr28}.  The eigenbases of complementary observables are \emph{mutually unbiased} in the sense that the transition
probabilities  across their basis states are all equal \cite{Schw60, Ivan81,WootF89,DurtEBZ10}. Conversely, given two  bases that are mutually unbiased, one can construct two complementary observables. Complementary observables and mutually unbiased bases (MUB) are two faces of the same coin.
In a $d$-dimensional Hilbert space, each MUB contains at most
$d+1$
bases \cite{WootF89}; the MUB is complete if the upper bound is attained. In the rest of the paper by a MUB we shall mean a complete set of mutually unbiased bases.

All MUB known so far only occur when the dimension $d$ is  a prime power
\cite{Ivan81,WootF89, BandBRV02, DurtEBZ10}.  Almost all of them can be equivalently constructed from \emph{stabilizer states} \cite{Gott97the,LawrBZ02, BandBRV02,GodsR09, Kant12}, which are simultaneous eigenstates of Heisenberg-Weyl (HW)  displacement operators (also known as generalized Pauli operators). They are called stabilizer MUB henceforth. These  MUB are of fundamental interest in diverse research areas, including but not limited to  quantum information, quantum foundations, and combinatorics.
For example, each stabilizer MUB can be used to define a family of discrete Wigner functions \cite{GibbHW04}. Interestingly, the set of pure states that have nonnegative Wigner functions, often dubbed as classical states,  happen to be the stabilizer states in the MUB, so the unitary transformations that preserve nonnegativity of the family of Wigner functions  happen to be the symmetry transformations of the MUB \cite{BengE05,CormGGP06}. Therefore, any progress in understanding the symmetry of stabilizer MUB would potentially benefit a number of research fields.

A MUB is group covariant if it can be generated from a single state---the \emph{fiducial state}---by a group composed  of unitary operators; the MUB is \emph{sharply covariant} if the  group  (modulo phase factors) can be chosen to have order $d(d+1)$.
Sharply covariant MUB, if they exist, would be  most appealing to  theoretical studies and practical applications. On the one hand,  they can be generated most efficiently in practice. On the other hand, they can be labeled naturally by group elements as phase point operators, which is crucial to phase space representation of quantum mechanics. Given the great variety of stabilizer MUB \cite{Kant12}, it seems reasonable to expect that some of them  would be sharply covariant. Many MUB are group covariant,
but the generating groups  usually have orders much larger than  $d(d+1)$.  For example, when $d$ is a prime, the stabilizer MUB can be generated by the Clifford group, which has order $d^3(d^2-1)$ \cite{ApplDF14}.

Here we show that only two stabilizer MUB are sharply covariant up to unitary equivalence, contrary to naive expectation.
Moreover, the conclusion still holds even if antiunitary transformations are taken into account.
Our study reveals a peculiar characteristic of all known MUB barring  a few exceptions. It also exposes
a sharp contrast between MUB and another elusive discrete structure known as  symmetric informationally complete measurements (SICs) \cite{Zaun11, ReneBSC04,ScotG10, Zhu12the, ApplFZ15G}, all known examples of which are sharply covariant.

In prime dimension $p$, the HW group $\hw$ is generated by the phase operator $Z$ and cyclic shift operator $X$ (together with scalar $\rmi$ when $p=2$),
\begin{equation}
Z|e_r\rangle=\omega^r|e_r\rangle, \qquad X|e_r\rangle=
|e_{r+1}\rangle, \label{eq:HW}
\end{equation}
where $\omega=\mathrm{e}^{2\pi \mathrm{i}/p}$,
$r\in \bbF_p$, and $\bbF_p$  is the field
of integers modulo $p$. The (multipartite) HW group in prime power dimension $q=p^n$ is the tensor power of $n$ copies of the HW group in dimension $p$.
The elements of the HW group are called \emph{displacement operators}. Up to phase factors they  can be labeled by  vectors $\mu$ of length $2n$ over $\bbF_p$ as
$D_{\mu}=\prod_{j=1}^n X_j^{\mu_j} Z_j^{\mu_{n+j} }$, where $Z_j$ and $X_j$ are phase operator and cyclic shift operator of the $j$th party.
These operators satisfy the commutation relation
$D_{\mu}D_{\nu} D_{\mu}^\dag D_{\nu}^\dag  =\omega^{\langle\mu,\nu\rangle}$,
where $\langle\mathbf{\mu},\mathbf{\nu}\rangle=\mu^\rmT J\nu$ is the symplectic product with
$J=\bigl(\begin{smallmatrix}0_n &-1_n\\ 1_n& 0_n
\end{smallmatrix}\bigr)$.
Two displacement operators $D_{\mu}$ and $D_{\nu}$ commute if and only if the corresponding symplectic product $\langle\mathbf{\mu},\mathbf{\nu}\rangle$ vanishes.  The vectors $\mu$ together with the symplectic product form a symplectic space  of dimension $2n$.

A stabilizer basis is the common eigenbasis of a maximal abelian subgroup of the HW group, where a maximal abelian subgroup is an abelian subgroup of order $q$ modular phase factors.
When $q$ is a prime, there are $q+1$ stabilizer bases, and the stabilizer MUB is unique.
Otherwise,  many different  MUB can be constructed from stabilizer bases \cite{Kant12}.
Most existing literature on MUB has focused on a particular construction based on field extension \cite{WootF89, BandBRV02, DurtEBZ10}, and  little is known about stabilizer MUB in general. Our work shall fill this gap.

The Clifford group $\Cli$ is the normalizer of the HW group that  is composed of all unitary operators that map displacement operators to displacement operators \cite{BoltRW61I,BoltRW61II,Gott97the,DehaM03}; the extended  Clifford group also contains antiunitary operators and is generated by the Clifford group and complex conjugation with respect to the computational basis. The importance of the Clifford group to the current study is reflected in the  observation: any unitary transformation between two stabilizer MUB is a Clifford unitary \cite{CormGGP06,Kant12}. This result can be extended to antiunitary transformations straightforwardly. In particular, the unitary and antiunitary symmetry transformations of a stabilizer MUB  belong to the extended Clifford group.\begin{theorem}
Any unitary or antiunitary transformation between two stabilizer MUB is a Clifford unitary or antiunitary.
\end{theorem}

Any Clifford unitary $U$ induces a symplectic transformation $F$ on the symplectic space, which labels the displacement operators. Conversely, given any symplectic matrix $F$, there exists a Clifford unitary $U_F$ that induces  $F$ \cite{BoltRW61I,BoltRW61II,DehaM03}. Actually, the  $q^2$  Clifford unitaries $D_\mu U_F$ for  $\mu \in \bbF_p^{2n}$ all induce the same  transformation. The quotient group $\pc/\phw$ ($\overline{G}$ denotes the group $G$ modulo phase factors) can be identified with the symplectic group $\Sp{2n}{p}$. When $p$ is odd, $\pc$ is also isomorphic to
the affine symplectic group $\ASp{2n}{p}$  \cite{BoltRW61I,BoltRW61II}.

Our main result can be formulated as follows.
\begin{theorem}\label{thm:MUBstabCov}
       All stabilizer MUB in dimensions 2 and 4 are sharply covariant;  all others are not.
\end{theorem}
\begin{remark}
In each dimension from 2 to 5, it is known that all MUB are equivalent \cite{BrieWB10}. So all MUB in dimensions 2 and 4 are sharply covariant, while no MUB in dimension 3 or 5 is sharply covariant.
\end{remark}

Let us first consider  prime dimension $p$.   In this case, the $p+1$ stabilizer bases   form a MUB, whose symmetry group coincides with the full Clifford group.
In the case of the qubit, all MUB are unitarily equivalent. Each   MUB forms an octahedron when represented on the Bloch sphere. The Clifford group corresponds to the (proper) octahedron group.  The MUB is sharply covariant with respect to each of the four  order-6 subgroups in the Clifford group.

In odd prime dimension $p$,   suppose the stabilizer MUB is sharply covariant with respect to   $\overline{G}$, then $\overline{G}$ is a subgroup of the Clifford group of  order $p(p+1)$, so it contains an element of order 2.
The Clifford group is isomorphic to $\SL{2}{p}\ltimes \bbF_p^2$, and up to phase factors all order-2 elements are conjugated to the parity operator, which  induces the symplectic transformation $-\diag(1,1)$~\cite{Zhu10}.
Since  the parity operator leaves all stabilizer bases  invariant,  the  stabilizer of each basis has order divisible by 2,
in contradiction with the requirement that  it has order $p$. So no stabilizer MUB in   dimension $p$ is sharply covariant.

To prove \thref{thm:MUBstabCov} in general, we need to  introduce several tools. A measurement $\{\Pi_j\}$ is informationally complete (IC) if the outcomes  $\Pi_j$ span the operator space~\cite{Zhu12the}. It is covariant with respect to the group $\overline{G}$  if it can be generated by $\overline{G}$ from one of the outcomes.
\begin{lemma}\label{lem:ICirr}
Suppose $\{\Pi_j\}$ is an IC measurement that is covariant with respect to  $\overline{G}$. Then $\overline{G}$ is irreducible.
\end{lemma}
\begin{proof}
Suppose on the contrary that $\overline{G}$ is reducible. Let $P$ be the projector onto a nontrivial invariant subspace of $\overline{G}$. Then $\tr(P\Pi_j)$ is independent of $j$ since $\overline{G}$ acts transitively on the $\Pi_j$. Given that $\Pi_j$ is IC, it follows  that $P$ is proportional to the identity, in contradiction with the assumption. Therefore, $\overline{G}$ is irreducible.
\end{proof}

The order of any irreducible  cyclic subgroup of $\Sp{2n}{p}$ is a divisor of $q+1$; all such subgroups of a given order are conjugated to each other.   Those cyclic subgroups  of order $q+1$, which  are always irreducible,
are called Singer cyclic subgroups and their generators called \emph{Singer cycles} \cite{Hupp70, Bere00}. The centralizer of any irreducible cyclic subgroup is  a singer cyclic group. The centralizer of a Singer cyclic group is itself, and the normalizer  has order $2n(q+1)$ \cite{Shor92book}.
Let $a,b$ be positive integers with $b>1$. A prime $r$ dividing  $b^a-1$ is a \emph{Zsigmondy prime} (also known as primitive prime divisor) \cite{Zsig1892,Roit97} if $r$ does not divide $b^j-1$ for $j=1, 2,\ldots, a-1$.
It is known that $b^a-1$ has a Zsigmondy prime except when $(b,a)=(2,6)$, or $a=2$ and $b+1$ is a power of 2.
A cyclic subgroup $R$ of $\Sp{2n}{p}$ of prime order $r$ is irreducible if and only if
$r$ is a Zsigmondy prime of $p^{2n}-1$. In that case, the subgroup is called a Zsigmondy cyclic subgroup, and any generator is called a \emph{Zsigmondy cycle}.

Before proving \thref{thm:MUBstabCov}, we need to generalize the concepts of Singer cycles and Zsigmondy cycles to the Clifford group. A Clifford unitary $\overline{U}$ in dimension $q=p^n$ is called a \emph{Singer unitary} (\emph{Zsigmondy unitary}) if its order is equal to $q+1$ (a Zsigmondy prime of $p^{2n}-1$); the group generated by such $\overline{U}$ is called a Singer unitary group (Zsigmondy unitary group).  Singer and Zsigmondy unitaries are useful not only in   establishing our main result but also in studying  MUB cycling problem and discrete Wigner functions \cite{Zhu15S}.
\begin{lemma}\label{lem:Singer}
\begin{enumerate}\itemsep-0.5ex
	
  \item  A Clifford unitary is a Singer unitary if and only if its induced symplectic transformation is a Singer cycle.

  \item All Singer unitary subgroups of the Clifford group $\pc$ are conjugated to each other.

  \item The centralizer of a Singer unitary group is itself.

  \item The normalizer of a Singer unitary group has order $2n(q+1)$.

   \item   $|\tr(U)|^2=1$ for any element $U$ in a Singer unitary group that is not proportional to the identity.

   \item Eigenvalues of a Singer unitary are nondegenerate.
\end{enumerate}
\end{lemma}

\begin{proof}
Let  $V$ be a Clifford unitary with induced symplectic transformation $S$; then  $D_\nu V D_\nu^\dag \propto D_{(1-S) \nu} V$.
If $V$ is a Singer unitary,  then $S$ has order $q+1$ and is thus a Singer cycle.  Conversely, if $S$ is a Singer cycle,
then all its eigenvalues (in the extension field $\bbF_{q^2}$) are different from 1, so  $1-S$ is invertible; that is, the range of $1-S$ is the whole symplectic space. So  $\overline{V}$ does not commute with any (nontrivial) displacement operator. The group generated by $\overline{V}$ cannot contain any displacement operator and is thus isomorphic to the group generated by $S$. In particular, $\overline{V}$ has order $q+1$ and is a Singer unitary.

Above analysis  shows that $\overline{D_\mu} \overline{V}$ for $\mu\in \bbF_p^{2n}$ and  $\overline{V}$ are conjugated to each other; that is, all Singer unitaries that induce the same symplectic transformation are conjugated to each other. Now statement~2 follows from the observation that
 all Singer cyclic groups in  $\Sp{2n}{p}$ are conjugated to each other.

Statement~3 holds because the centralizer of a  Singer unitary group  does not contain any nontrivial displacement operator, so  its order is no larger than  the order $q+1$ of the centralizer of the corresponding Singer cycle.

Statement~4 follows from the observation that the normalizer of a Singer unitary group has the same order as the normalizer  of a Singer cycle, noting that the number of Singer unitary groups is $q^2$ times the number of Singer cyclic groups in $\Sp{2n}{p}$.

Statement~5 follows from the observation that  $\overline{D_\mu}\overline{U}$  for  $\mu\in \bbF_p^{2n}$ and  $\overline{U}$ are conjugated to each other (as in the case $\overline{U}$ is a Singer unitary) and that
$\sum_\mu |\tr(D_\mu U)|^2=q^2$ since the HW group constitutes a unitary error basis.

 As a consequence of  statement 5, the sum  of squared multiplicities of inequivalent irreducible components of a  Singer unitary group is given by $(q^2+q)/(q+1)=q$, from which statement~6 follows.
\end{proof}

\begin{lemma}\label{lem:PPD}
All Zsigmondy unitary groups of a given order $r$  are conjugated to each other in the Clifford group; the centralizer of each one is a Singer unitary group; $|\tr(U)|^2=1$ for any  Zsigmondy unitary $U$.
\end{lemma}

\begin{proof}
Let $F$ be the symplectic transformation  induced by $\overline{U}$; then  $F$ is a Zsigmondy cycle of $\Sp{2n}{p}$ and is thus  a power of a Singer cycle, which implies  that  $1-F$ is invertible. Therefore,  $\overline{D_\mu} \overline{U}$ for $\mu\in \bbF_p^{2n}$ and  $\overline{U}$ are conjugated to each other. It follows that $|\tr(U)|^2=1$. In addition, all Zsigmondy unitary subgroups of order $r$ are conjugated to each other, given that the same holds for Zsigmondy cyclic subgroups of  $\Sp{2n}{p}$. Consequently, each Zsigmondy unitary group is contained in  a Singer unitary group and its centralizer has order at least $q+1$. On the other hand, the order of the centralizer of $\overline{U}$ is no larger than the order $q+1$ of  the centralizer of $F$ in $\Sp{2n}{p}$ since $\overline{U}$ does not commute with any nontrivial displacement operator. It follows that the centralizer  has order $q+1$ and is a Singer unitary group.
\end{proof}

Now we are ready to prove \thref{thm:MUBstabCov} in the remaining case  $q=p^n$ with $n\geq2$.
\begin{proof}[Proof of \thref{thm:MUBstabCov}]
Suppose a stabilizer MUB in dimension $q$  is sharply covariant with respect to $\overline{G}$. Then $\overline{G}\in \pc$ and  $|\overline{G}|=q(q+1)$. In addition, $\overline{G}$ is irreducible according to \lref{lem:ICirr}. Let $\overline{T}=\overline{G}\cap \phw$ and $Q=\overline{G}/\overline{T}$. Then $\overline{T}$ is an elementary abelian $p$-group, which can be identified as a subspace of  $\bbF_p^{2n}$, while $Q$ can be identified as a subgroup of $\Sp{2n}{p}$ that stabilizes the subspace.

If $q\neq 8$, then $|\overline{G}|$ and $|Q|$ are divisible by a Zsigmondy prime $r$ of $p^{2n}-1$. Therefore, $Q$  is irreducible on $\bbF_p^{2n}$, which implies that $\overline{T}$ is either trivial or isomorphic to $\bbF_p^{2n}$. The latter possibility cannot happen since the order of $\overline{T}$ is at most $q$.
Let $\overline{U}\in\overline{G}$ be a Zsigmondy unitary of order  $r$  and  $C_{\overline{G}}(\overline{U})$ its  centralizer. According to \lref{lem:PPD}, $|\tr(U)|^2=1$; in addition,
$C_{\overline{G}}(\overline{U})$ is a subgroup of a Singer unitary group, so its order can be written as  $(q+1)/a$ with $a$ a divisor of $q+1$; that is, $\overline{U}$ has $a q$ conjugates in $\overline{G}$.  The sum of squared multiplicities of inequivalent irreducible components of $\overline{G}$  satisfies
 \begin{equation}\label{eq:IrrC}
\frac{1}{q(q+1)}\sum_{\overline{V}\in \overline{G}} |\tr(V)|^2\geq \frac{q^2+aq}{q(q+1)}\geq 1.
\end{equation}
Since  $\overline{G}$ is irreducible, it follows  that  $a=1$, so $C_{\overline{G}}(\overline{U})$ is a Singer unitary  group. The number of conjugates of $\overline{U}$ that are contained in $C_{\overline{G}}(\overline{U})$ is  no larger than the index of $C_{\overline{G}}(\overline{U})$ in its normalizer, which is equal to
$2n$ according to \lref{lem:Singer}. In addition, $|\tr(V)|^2=1$ for any nontrivial $\overline{V}\in C_{\overline{G}}(\overline{U})$.
So  \eref{eq:IrrC} can be strengthened as
 \begin{equation}\label{eq:IrrC2}
 \frac{1}{q(q+1)}\sum_{\overline{V}\in \overline{G}} |\tr(V)|^2\geq \frac{q^2+q+(q-2n)}{q(q+1)}\geq 1.
 \end{equation}
The lower bound can be saturated  only when $2n=q=p^n$, that is, $p^n=4$ (assuming $n\geq2$). Therefore, no stabilizer MUB is sharply covariant except possibly for dimensions  4 and 8. It turns out the same conclusion applies to
dimension 8, as shown in the appendix.

In dimension 4, the Clifford group has order 11520. Its quotient over the HW group is isomorphic to  the symplectic group $\Sp{4}{2}$, which in turn  is isomorphic to the symmetric group on six letters. Calculation shows that there are six stabilizer MUB \cite{KlimRBS07}, any permutation  of which can be realized  by Clifford unitary transformations.
The symmetry group of each MUB has order 1920, its quotient over the HW group is isomorphic to the symmetric group on five letters, and it can realize any permutation among the
five bases in the MUB. In addition to having this remarkable symmetry, the  stabilizer  MUB in dimension 4 turns out to be the only exception beyond qubit
that  is sharply covariant. Indeed, the MUB is sharply covariant with respect to the normalizer of each Sylow 5-subgroup of the symmetry group. One of the groups is generated by
\begin{equation}\label{eq:SingerNord4}
U_1\rep\frac{1}{2}\begin{pmatrix}\rmi & 1 & \rmi &-1\\
 \rmi & -1 &\rmi &1\\
 \rmi& 1& -\rmi &1\\
 \rmi&-1& -\rmi&-1
\end{pmatrix},\quad
U_2\rep\begin{pmatrix}0 & 0 &0 &\rmi\\
\rmi & 0 &0 &0\\
0& -1& 0 &0   \\
0&0& 1&0
\end{pmatrix},
\end{equation}
which satisfy $U_1^5=1$, $U_2^4=1$, $U_2U_1U_2^\dag= U_1^2$. Note that $U_1$ is simultaneously a Singer unitary, a Zsigmondy unitary, and  a Hadamard matrix.
Any state in the computational basis is a fiducial state for a stabilizer MUB.
The unitary  $U_1$ cycles the five bases, while $U_2$ cycles the four states in the computational basis. In total, each stabilizer MUB is sharply covariant with respect to  96 unitary  groups, all of which are conjugated to each other in the symmetry group of the MUB.
\end{proof}

According to Wigner theorem, any symmetry transformation of the quantum state space is either unitary or antiunitary.
Is there any
other sharply covariant stabilizer MUB if antiunitary transformations are taken into account? In dimension 2, any order-6 antiunitary transformation in the extended Clifford group (corresponding to the product of inversion and an order-3 rotation in the Bloch-sphere representation)  cycles not only the three bases but also all six basis states. So  each MUB is sharply covariant with respect to four antiunitary groups in addition to four unitary  groups. In dimension 4,  each stabilizer MUB is sharply covariant with respect to 96 antiunitary groups in addition to 96 unitary  groups.
To be specific, the group generated by $U_1$ and $U_2$ in \eref{eq:SingerNord4}
 is centralized by the Clifford antiunitary $U_3=(\sigma_y\otimes 1)\hat{K}$ up to phase factors, where $\sigma_y=\sy$ is a Pauli matrix and $\hat{K}$ is the complex conjugation operator. The corresponding MUB is  sharply covariant with respect to the antiunitary group generated by $U_1$ and $U_2U_3$, as well as its conjugates.
\begin{theorem}\label{thm:MUBstabCov2}
	 No  stabilizer MUB other than those in  dimensions 2 and 4 is sharply covariant with respect to any group composed of unitary and antiunitary operators.
\end{theorem}
\begin{proof}
	Suppose on the contrary that there is a stabilizer MUB in dimension $q=p^n$ with $q\neq 2,4$ that is sharply covariant with respect to $\overline{G}$.
	 Then $\overline{G}$ has order $q(q+1)$ and it must contain some antiunitary operators according to \thref{thm:MUBstabCov}.  Observing  that \lref{lem:ICirr} also applies to antiunitary groups, we conclude that  $\overline{G}$ is irreducible.
	 Let $\overline{H}$ be the subgroup of $\overline{G}$ that is composed of unitary operators. Then $\overline{H}$ has index 2 and is normal  in $\overline{G}$.  Consequently,  $\overline{H}$ is either irreducible or has two irreducible components of equal degree. The first possibility cannot happen because the order of any irreducible group in dimension $q$ is at least $q^2$, while $\overline{H}$ has order $q(q+1)/2$. Consequently,  $p$ must equal 2.
	
	 If $q\neq 8$, then $|\overline{H}|$ is divisible by a Zsigmondy prime $r$ of $2^{2n}-1$.
	 Let $\overline{U}\in\overline{H}$ be a Zsigmondy unitary of order~$r$; then the order of its centralizer  $C_{\overline{H}}(\overline{U})$ is $(q+1)/a$ for some divisor $a$ of $q+1$.  According to a similar reasoning that leads to \eref{eq:IrrC}, we have
	 \begin{equation}\label{eq:IrrC3}
	 2=\frac{2}{q(q+1)}\sum_{\overline{V}\in \overline{H}} |\tr(V)|^2\geq \frac{2q^2+aq}{q(q+1)},
	 \end{equation}
	 where the first equality follows from the observation that the two irreducible components of $\overline{H}$ are inequivalent, as proved in the appendix.
	 Since $a$ is odd, it follows that $a=1$ and  $C_{\overline{H}}(\overline{U})$ is a Singer unitary group. As in the proof of \thref{thm:MUBstabCov}, \eref{eq:IrrC3} can now be strengthened as
	 \begin{equation}\label{eq:IrrC4}
	2= \frac{2}{q(q+1)}\sum_{\overline{V}\in \overline{H}} |\tr(V)|^2\geq \frac{2q^2+q+2(q-2n)}{q(q+1)}.
	 \end{equation}
	 The inequality can never hold  when  $n>4$. Therefore, no stabilizer MUB is sharply covariant when $q$ is odd or $q>16$ even if we allow antiunitary transformations.
	
	 To complete the proof, it remains to consider stabilizer MUB in dimensions 8 and 16. The former is settled in the appendix. In dimension 16, each Singer unitary group has order 17, and its normalizer in the extended Clifford group has order 272. It follows that $\overline{G}$ is the normalizer of a Singer unitary group. Calculation shows that the normalizer  has two irreducible components, so  no stabilizer MUB in dimension  16 is sharply covariant.
\end{proof}

In summary we have introduced the concept of sharply covariant MUB, which are distinguished from generic group covariant MUB by smallest possible generating groups. Although there are a great variety of stabilizer MUB, we proved that
only two of them  are sharply covariant, and the conclusion remains intact even if antiunitary transformations are taken into account. Our study provides valuable insight
on  the symmetry  of stabilizer MUB, which may help understand a number of topics in quantum information and foundations, for example, those classicality-preserving unitary  transformations  with respect to discrete Wigner functions based on stabilizer MUB.
 Our work also  reveals a deep structure distinction between MUB and SICs, which is of intrinsic interest to  research on quantum geometry. Furthermore, the ideas and tools introduced in the course of study are useful to studying other problems related to MUB, SICs, discrete Wigner functions, and Clifford groups etc.

\section*{Acknowledgments}
The author is grateful to Nick Gill for introducing the concepts of Singer cycles and Zsigmondy primes and to Masahito Hayashi for comments. This work is supported in part by Perimeter Institute for Theoretical Physics. Research at Perimeter Institute is supported by the Government of Canada through Industry Canada and by the Province of Ontario through the Ministry of Research and Innovation.

\appendix
\section{Appendix A: Stabilizer MUB in dimension 8}
In this appendix we show that no stabilizer MUB in dimension 8 is sharply covariant with respect to any group composed of unitary and antiunitary operators. This case requires special treatment because $2^6-1$ has no Zsigmondy prime.
Calculation shows that there are 960 stabilizer MUBs, all of which can be transformed into each other under the Clifford group. The symmetry group of each MUB has order 96768, and its quotient $B$ over the HW group has order 1512. The group $B$ can be identified as the normalizer in $\Sp{6}{2}$ of an extension field type subgroup isomorphic to $\SL{2}{8}$ \cite{Kant12}, which has order 504 and index 3 in $B$. Each Sylow 3-subgroup $E$ of $B$ is an extraspecial group of order $27$ and exponent 9;  it has a unique order-9 subgroup $P$ that is not cyclic. It acts transitively on the nine bases of the MUB; the stabilizer of each basis under this action has order 3 and is always contained in $P$.

Suppose the   MUB  is sharply covariant with respect to $\overline{G}$; then $\overline{G}$ is irreducible according to \lref{lem:ICirr}.  Let $\overline{H}$ be the subgroup of $\overline{G}$ composed of unitary operators; then $\overline{H}$ is either identical with $\overline{G}$ or has index 2 and it is either irreducible or has two irreducible components of equal degree accordingly. Let $\overline{T}=\overline{H}\cap \phw$,  $Q=\overline{H}/\overline{T}$, and  $A$  any Sylow 3-subgroup of $Q$, then $A$ is an order-9 subgroup of $E$. If $A$ is not cyclic, then the stabilizer of each basis has order divisible by 3, so that $\overline{G}$ is not transitive on the nine bases of the MUB, in contradiction with the assumption. Otherwise,
$A$ is a Singer cyclic subgroup of $\Sp{6}{2}$, and $\overline{H}$ contains a Singer unitary subgroup, say $\overline{S}$.

Since  $\overline{H}$ has order either 72 or 36, while the normalizer of $\overline{S}$ in the Clifford group has order 54 according to \lref{lem:Singer},  $\overline{S}$ can not be normal in $\overline{H}$.
On the other hand,
 $|\tr(U)|^2=1$ for any nontrivial element $\overline{U}$ in a Singer unitary group.   So the sum of squared multiplicities of inequivalent irreducible components of $\overline{H}$  satisfies
\begin{equation}\label{eq:IrrC5}
\frac{1}{|\overline{H}|}\sum_{\overline{U}\in \overline{H}} |\tr(U)|^2> \frac{q^2+q}{|\overline{H}|}= \frac{|\overline{G}|}{|\overline{H}|}.
\end{equation}
This inequality cannot hold if $\overline{H}=\overline{G}$. If  $\overline{H}$  has index 2 in  $\overline{G}$ and thus two irreducible components, then the two irreducible components are necessarily inequivalent, given that $\overline{H}$ contains a Singer unitary, whose eigenvalues are nondegenerate. So the left hand side of \eref{eq:IrrC5} equals 2, and the inequality cannot hold either. It follows that no stabilizer MUB in dimension 8  is sharply covariant.

\section{Appendix B: A technical lemma}
\begin{lemma}\label{lem:PPDineq}
	Suppose  $\overline{H}$ is a subgroup of the Clifford group in dimension $q=2^n$ that contains a Zsigmondy unitary.
If $\overline{H}$ has two irreducible components of the same degree $2^{n-1}$, then the two irreducible components are inequivalent.
\end{lemma}
\begin{proof}
	Let $\overline{U}\in \overline{H}$ be a Zsigmondy unitary of order $r$, then $U$ is a power of a Singer unitary, say $V$. Since all eigenvalues of $V$ are nondegenerate, with a suitable choice of the phase factor if necessary, we may assume that these eigenvalues are $\xi^j$ for $j=1,2,\ldots,q$ where $\xi$ is a primitive $(q+1)$th root of unity. Let $\zeta$ be a primitive $r$th root of unity, then  the eigenvalues of $U$ are $\zeta^j$ for $j=0,1,\ldots,r-1$, each of multiplicity $(q+1)/r$ except for the eigenvalue $1$ of multiplicity  $[(q+1)/r]-1$. It follows that $\tr(U)=-1$. Let $t_1$ and $t_2$ be the trace of $U$ within the two irreducible components of $\overline{H}$, respectively. Then both $t_1$ and $t_2$ are linear combinations of $r$th roots of unity with integer coefficients. If the two irreducible components of $\overline{H}$ are equivalent, then $t_1=t_2$. Consequently, $\tr(U)=t_1+t_2=2t_1=-1$, so that $2t_1+1=0$. Expressing  $t_1$ as a linear combination of powers of $\zeta$, we get $\sum_{j=0}^{r-1} c_j\zeta^j=0$, where $c_j$ are integers, all of which are even except for $c_0$. Therefore,
	$\zeta$ is a root of the polynomial $f(x)=\sum_{j=0}^{r-1} c_jx^j$, which implies that $f(x)$ is a multiple of the minimal polynomial $\sum_{j=0}^{r-1} x^j$
	of $\zeta$, which is impossible. This contradiction shows that the two  irreducible components of $\overline{H}$ are inequivalent.
\end{proof}

\bibliographystyle{apsrev4-1}

\bibliography{all_references}

\end{document}